\documentclass[runningheads,a4paper]{llncs}

\usepackage{amssymb}
\usepackage{graphicx}
\usepackage{amsmath}
\usepackage{url}
\usepackage[english]{babel}
\usepackage{enumitem}
\usepackage{enumerate}
\usepackage{color}
\usepackage[T1]{fontenc}
\usepackage{subfigure}
\usepackage{dsfont}

\urldef{\mailsa}\path|{asoodehshahab, fady, linder|\urldef{\mailsb}\path||
\urldef{\mailsc}\path|}@mast.queensu.ca|
\newcommand{\keywords}[1]{\par\addvspace\baselineskip
\noindent\keywordname\enspace\ignorespaces#1}
\def\R{\mathbb{R}}

\def\fin{{<\infty}}

\def\eps{\varepsilon}

\def\E{{\mathbb E}}

\def\I{{\mathcal I}}

\def\C{{\mathcal C}}

\def\X{{\mathcal X}}
\def\Y{{\mathcal Y}}

\def\N{{\mathcal N}}

\def\Z{{\mathcal Z}}

\def\sBer{{\mathsf{Bernoulli}}}

\def\sM{{\mathsf {sENSR}}}

\def\sW{{\mathsf {wENSR}}}
\def\sG{{\mathsf G}}
\def \var {{\mathsf {var}   }}
\def \mmse {{\mathsf {mmse}   }}

\newcommand{\markov}{\mathrel\multimap\joinrel\mathrel-%
\mspace{-9mu}\joinrel\mathrel-}
\begin{document}

\mainmatter  

\title{Almost Perfect Privacy for Additive Gaussian Privacy Filters}

\titlerunning{Almost Perfect Privacy in Additive Gaussian Privacy Filter}

%
%
\author{Shahab Asoodeh%
\thanks{This work was supported in part by NSERC of Canada.}%
\and Fady Alajaji\and Tam\'{a}s Linder}
\authorrunning{Almost Perfect Privacy for Additive Gaussian Privacy Filters}

\institute{Department of Mathematics and Statistics, Queen's University\\
Jeffery Hall, 48 University Ave., Kingston, ON, Canada\\
\mailsa
\mailsc}

%
%

\toctitle{Lecture Notes in Computer Science}
\tocauthor{Authors' Instructions}
\maketitle

\begin{abstract}
We study the maximal mutual information about a random variable $Y$ (representing non-private information) displayed through an additive Gaussian channel when guaranteeing that only $\eps$ bits of information is leaked about a random variable $X$ (representing private information) that is correlated with $Y$. Denoting this quantity by $g_\eps(X,Y)$, we show that for perfect privacy, i.e., $\eps=0$, one has $g_0(X,Y)=0$ for any pair of absolutely continuous random variables $(X,Y)$ and then derive a second-order approximation for $g_\eps(X,Y)$ for small $\eps$. This approximation is shown to be related to the strong data processing inequality for mutual information under suitable conditions on the joint distribution $P_{XY}$. Next, motivated by an operational interpretation of data privacy, we formulate the privacy-utility tradeoff in the same setup using estimation-theoretic quantities and obtain explicit bounds for this tradeoff  when $\eps$ is sufficiently small using the approximation formula derived for $g_\eps(X,Y)$.
\end{abstract}

\section{Introduction}

The ever increasing growth of social networks has brought major challenges in terms of data privacy. This paper focuses on a privacy problem which is relevant for users or designers of social networks: the trade-off between data privacy and customized services performance. On the one hand, users want their private data to remain secret, and on the other hand, they also desire to benefit from customized services that require personal information in order to function properly. In this context, it is reasonable to assume that the user has two kinds of data: private data such as passport numbers, credit cards numbers, etc; and non-private data such as gender, age, etc. In general, private and non-private data are correlated. Thus, it is possible that enough non-private data discloses a non-negligible amount of private data. Therefore, it is necessary to develop techniques to provide/store personal data (user's point of view/designer's point of view) that yield the best customized services performance without compromising privacy. The goal of these techniques is to provide displayed data that will be used by customized services which contains as much non-private data as possible while revealing as little private data as possible. Also, for security reasons, the displayed data has to be produced using only non-private data. In general, this implies that the displayed data should be a randomized version of the non-private data.

To formulate this problem, we need to specify a privacy function and a utility function that respectively measure the amount of private and non-private data \emph{leaked} into the displayed data. The authors of this paper recently suggested in \cite{Asoodeh_Allerton} to use mutual information as the measure of both utility and privacy. Let $X$ and $Y$ denote the private and non-private data, respectively.
The \emph{rate-privacy function} $g^{\mathsf{dis}}_\eps(X,Y)$ for discrete random variables $X$ and $Y$ having finite alphabets $\X$ and $\Y$, respectively is defined for any $\eps\geq 0$ as the privacy-utility tradeoff
 \begin{equation}\label{gEpsilon}
  g^{\mathsf{dis}}_\eps(X,Y):=\max_{\substack{P_{Z|Y}:X\markov Y\markov Z,\\I(X; Z)\leq \eps}} I(Y;Z),
\end{equation}
where the auxiliary random variable $Z$ is the privacy-constrained displayed data and $X\markov Y\markov Z$ denotes that $X$, $Y$, and $Z$ form a Markov chain in this order. The channel $P_{Z|Y}$ is called the \emph{privacy filter}. It is shown in \cite{Asoode_submitted} that $g^{\mathsf{dis}}_\eps(X,Y)$ is in fact a corner point of an
outer bound on the achievable region of the "dependence dilution" coding problem which provides an information-theoretic operational interpretation. It is also shown that if  the channel from $Y$ to $X$ displays certain symmetry properties, then $g^{\mathsf{dis}}_\eps(X,Y)$ can be calculated in closed form.  For instance, if $P_{X|Y}$ is a binary symmetric channel (BSC) and $Y\sim\sBer(0.5)$, then $g^{\mathsf{dis}}_\eps(X,Y)=\frac{\eps}{I(X;Y)}$.

As a more practical and operational notion of privacy, estimation-theoretic formulations of privacy are introduced in \cite{Calmon_bounds_Inference} and \cite{Asoode_MMSE_submitted}.
In particular, Calmon et al.\ \cite{Calmon_bounds_Inference} studied the case where $X=Y$ and defined the utility by $\Pr(\hat{Y}(Z)=Y)$ where $\hat{Y}:\Z\to \Y$ is the Bayes decoding map  satisfying $I(Y; Z)\leq \eps$ for discrete $Y$. Motivated by \cite{Fawaz_Makhdoumi}, which suggested the use of maximal correlation $\rho_m^2(X, Z)$ to measure the privacy level between $X$ and $Z$, the authors in \cite{Asoode_MMSE_submitted} recently generalized this model to arbitrary discrete $X$ and $Y$, with the same utility function except that $Z$ is required to satisfy $\rho_m^2(X,Z)\leq \eps$. It was shown independently in \cite{Asoodeh_Allerton} and \cite{Calmon_fundamental-Limit} that if \emph{perfect privacy} is required, i.e., $Z$ must be statistically independent of $X$, then $Z$ is also independent of $Y$ unless the probability vectors $\{P_{Y|X}(\cdot|x):x\in \X\}$ are linearly dependent (in which case $Y$ is called \emph{weakly independent} of $X$, see \cite[Appendix II]{berger}). Hence, if $Y$ is not weakly independent of $X$, then $g^{\mathsf{dis}}_0(X,Y)=0$. Other formulations for privacy have appeared in \cite{Reed:1973,yamamotoequivocationdistortion,Lalitha_Forensics,Asoodeh_CWIT,t_closeness,Calmon_PhD_thesis}.

The setting where $(X,Y)$ is a pair of absolutely continuous random variables with $\X=\Y=\R$ is studied in \cite{Asoode_submitted} with both utility and privacy being measured by mutual information,  and in \cite{Asoode_MMSE_submitted}, where both utility and privacy are measured in terms of the minimum mean-squared error (MMSE). In both cases, it is assumed that the privacy filter is an additive Gaussian channel with signal-to-noise ratio (SNR) $\gamma\geq 0$, i.e.,
 \begin{equation}\label{Z_Gamma}
  Z=Z_\gamma:=\sqrt{\gamma}Y+N_\sG,
\end{equation}
 where $N_\sG\sim \N(0, 1)$ is independent of $(X,Y)$. In particular, the rate-privacy function \cite{Asoode_submitted} is defined as
\begin{equation}
\label{g_Epsilon_Con}
g_{\eps}(X,Y) := \max\limits_{\begin{smallmatrix}\gamma\geq0,\\I(X;Z_\gamma)\leq\eps\end{smallmatrix}} I(Y;Z_\gamma).
\end{equation}
Letting $\mmse(U|V)$ denote the MMSE of estimating $U$ by observing $V$ and letting $\var$ denote the variance, the estimation-theoretic privacy-utility tradeoff is defined in \cite{Asoode_MMSE_submitted} by the \emph{estimation noise-to-signal ratio} (ENSR):
\begin{equation}
\label{ENSR_Definition}
\sM_{\eps}(X,Y):=\min_{} \frac{\mmse(Y|Z_\gamma)}{\var(Y)},
\end{equation}
where the minimum is taken over all $\gamma\geq 0$ such that $\mmse(f(X)|Z_\gamma)\geq (1-\eps)\var(f(X))$ for any non-constant measurable function $f:\X\to \R$. Unlike $g_\eps(X,Y)$, $\sM_\eps(X,Y)$ has a clear operational interpretation; it is the smallest MMSE associated with  estimating $Y$ given $Z$ from which no non-degenerate function $f$ of $X$ can  be estimated efficiently. This notion is related to \emph{semantic security} \cite{Goldwasser1984270} in cryptography. An encryption mechanism is said to be semantically secure if the adversary's advantage for correctly guessing any function of the private data given an observation of the mechanism's
output (i.e., the ciphertext) is required to be negligible.  As opposed to the discrete case, perfect privacy is achieved if and only if $\gamma=0$, which gives rise to $g_0(X,Y)=0$ (or equivalently $\sM_0(X,Y)=1$) for any absolutely continuous $(X,Y)$.


\subsection{Contributions}
In this work, we investigate the "almost" perfect privacy regime, that is, when $\eps>0$ is close to zero and derive a second-order approximation for $g_\eps(X,Y)$ (Corollary 2). We also obtain the first and second derivatives of the mapping $\eps\mapsto g_\eps(X,Y)$ for $\eps\in [0, I(X;Y))$ (Theorem 1). For a pair of Gaussian random variables $(X,Y)$, an expression for $g_\eps(X,Y)$ is derived (Example 1) and it is shown that the optimal filter has SNR equal to $\frac{2^{2\eps}-1}{1-2^{-2(I(X;Y)-\eps)}}$ for all $\eps< I(X;Y)$ and the SNR is infinity if $\eps\geq I(X;Y)$. Functional properties of the map $\eps\mapsto g_\eps(X,Y)$ are obtained (Proposition 1); in particular, it is shown than although the map $\eps\mapsto g^{\mathsf{dis}}_\eps(X,Y)$ is concave \cite{Asoode_submitted}, the map $\eps\mapsto g_\eps(X,Y)$ is neither convex nor concave, and is infinitely differentiable (Corollary 1). Using a recent result on the strong data processing inequality by Anantharam et al.\ \cite{anantharam}, a lower bound is obtained for $g_\eps(X,Y)$. Assuming $P_{Y|X}$ is a convolution with a Gaussian distribution, i.e., $Y=aX+M_\sG$, where $a\neq 0$ and $M_\sG\sim \N(0, \sigma_M^2)$ is independent of $X$, we obtain an inequality relating $\mmse(Y|Z_\gamma, X)$ to $\mmse(Y|Z_\gamma)$ from which a stronger version of Anantharam's data processing inequality is derived for our setup (Theorem 2).

One main result of this paper is to connect $g_\eps(X,Y)$ with $\sM_\eps(X,Y)$ in the almost perfect privacy regime when $X$ is Gaussian (Theorem 4). This connection allows us to translate the approximation obtained for $g_\eps(X,Y)$ to a lower bound for $\sM_\eps(X,Y)$.

\subsection{Preliminaries}

For a given pair of absolutely continuous random variables $(U, V)$, we interchangeably use $P_{UV}$ to denote the joint probability distribution and also the joint probability density function (pdf). The MMSE of estimating $U$ given $V$ is given by
 \begin{equation*}
  \mmse(U|V):=\E[\left(U-\E[U|V]\right)^2]=\E[\var(U|V)],
\end{equation*}
where $\var(U|V)=\E[(U-\E[U|V])^2|V]$. Guo et al.\ \cite{MMSE_Guo} proved the following so-called I-MMSE formula relating the input-output mutual information of the additive Gaussian channel $Z_\gamma=\sqrt{\gamma}Y+N_\sG$, where $N_\sG\sim\N(0,1)$ is independent of $X$, with the MMSE of the input given the output:
  \begin{equation}\label{I_MMSE}
    \frac{\text{d}}{\text{d}\gamma}I(Y;Z_\gamma)=\frac{1}{2}\mmse(Y|Z_\gamma).
  \end{equation}
 Since $X$, $Y$ and $Z_\gamma$ form the Markov chain $X\markov Y\markov Z_\gamma$, it follows that $I(X; Z_\gamma)=I(Y; Z_\gamma)-I(Y; Z_\gamma|X)$ and hence two applications of \eqref{I_MMSE} yields \cite[Theorem 10]{MMSE_Guo}
     \begin{equation}\label{I_MMSE2}
    \frac{\text{d}}{\text{d}\gamma}I(X; Z_\gamma)=\frac{1}{2}\left[\mmse(Y|Z_\gamma)-\mmse(Y|Z_\gamma, X)\right].
  \end{equation}
The second derivative of  $I(Y; Z_\gamma)$ and $I(X; Z_\gamma)$ are also known via the formula \cite{MMSE_Guo_Wu}
  \begin{equation}\label{Second_Derivative_MMSE}
    \frac{\text{d}}{\text{d}\gamma}\mmse(Y|Z_\gamma, X)=-\E[\var^2(Y|Z_\gamma, X)].
  \end{equation}

R\'{e}nyi \cite{Renyi-dependence-measure} defined the \emph{one-sided maximal correlation between $U$ and $V$} (see also \cite[Definition 7.4]{Calmon_PhD_thesis}) as
\begin{equation}\label{Eta_Definition}
  \eta^2_V(U):=\sup_g\rho^2(U, g(V))=\frac{\var(\E[U|V])}{\var(U)},
\end{equation} where $\rho(\cdot, \cdot)$ is the (Pearson) correlation coefficient,  the supremum is taken over all measurable functions $g$, and the equality follows from the Cauchy-Schwarz inequality. The law of total variance implies that \begin{equation}\label{Law-Total_variance_ETA_MMSE}
  \mmse(U|V)=\var(U)(1-\eta^2_V(U)).
\end{equation}
In an attempt of symmetrizing $\eta^2_V(U)$, R\'{e}nyi \cite{Renyi-dependence-measure} (see also \cite{gebelien} and \cite{sarmanov}) defined the \emph{maximal correlation}  as
\begin{equation}\label{Definition_MaximalCorrelation}
  \rho_m^2(U, V)=\sup_{f, g}\rho^2(f(U), g(V)).
\end{equation}
Comparing \eqref{Eta_Definition} with \eqref{Definition_MaximalCorrelation} reveals that
  \begin{equation}\label{Eta_Gaussian_inequality}
    \rho^2(X,Y)\leq\eta^2_X(Y)\leq \rho_m^2(X,Y).
  \end{equation}
  Clearly, unlike maximal correlation, $\eta_X(Y)$ is asymmetric, i.e., in general $\eta_X(Y)\neq \eta_Y(X)$, and hence according to R\'{e}nyi's postulates \cite{Renyi-dependence-measure}, it is not a "proper" measure of dependence. However, it turns out to be an appropriate measure of separability between private and non-private information in the almost perfect privacy regime (see Corollary~\ref{corollary_Approximation_g_eps}).
  On the other hand, maximal correlation satisfies all the R\'{e}nyi's postulates \cite{Renyi-dependence-measure}. In particular, it is symmetric and for  jointly Gaussian random variables $U$ and $V$ with correlation coefficient $\rho$, we have $\rho_m^2(U, V)=\rho^2$.

\section{Rate-Privacy Function for Additive Privacy Filters}
Consider a pair of absolutely continuous random variables $(X,Y)$ distributed according to $P_{XY}$. Let $X$ and $Y$ represent the \emph{private data} and the \emph{non-private data}, respectively. We think of $X$ as having fixed distribution $P_X$ and $Y$ being generated by the channel $P_{Y|X}$, predefined by nature. Now consider the setting where Alice observes $Y$ and wishes to describe it as accurately as possible to Bob in order to get a utility from him.  Due to the correlation between $Y$ and the private data $X$, Alice needs to provide Bob a noisy version $Z$ of $Y$, such that $Z$ cannot reveal more than $\eps$ bits of information about $X$. In fact, we assume that $Z$ is obtained via the privacy filter, $Z=Z_\gamma$ defined in \eqref{Z_Gamma}. The aim is to pick $\gamma\geq 0$ such that $Z_\gamma$ preserves the maximum amount of the information about $Y$ while satisfying the privacy constraint. The rate-privacy function $g_{\eps}(X,Y)$, defined in \eqref{g_Epsilon_Con}, quantifies the tradeoff between these conflicting goals \cite{Asoode_submitted}. 
Note that since $I(Y; Z_\gamma)=I(Y; Y+\frac{1}{\sqrt{\gamma}}N_\sG)$, we can interpret $\frac{1}{\gamma}$ as the noise variance. 
Due to the data processing inequality, one can restrict $\eps$ to the interval $[0, I(X;Y))$ in the definition of $g_\eps(X,Y)$ and consequently for any $\eps\geq I(X;Y)$ the optimal noise variance must be zero and hence $g_\eps(X,Y)=\infty$. The case where the displayed data is required to carry no information at all about $X$, i.e., where $\eps=0$, is often called \emph{perfect privacy}.

The maps $\gamma\mapsto I(Y; Z_\gamma)$ and $\gamma\mapsto I(X; Z_\gamma)$ are strictly increasing over $[0, \infty)$ \cite[Lemmas 16, 17]{Asoode_submitted} and hence there exists a unique $\gamma_\eps\in [0, \infty)$ such that $I(X; Z_{\gamma_\eps})=\eps$ and $g_\eps(X,Y)=I(Y; Z_{\gamma_\eps})$. This observation yields the following proposition.
\begin{proposition}\label{Proposition:Properties}
  For absolutely continuous random variables $(X,Y)$, we have
  \begin{itemize}
    \item[1.] The map $\eps\mapsto \gamma_\eps$ is strictly increasing and continuous, and it satisfies $\gamma_0=0$ and $\gamma_{I(X;Y)}=\infty$.
    \item[2.] The map $\eps\mapsto g_\eps(X,Y)$ is non-negative, increasing and, continuous on $[0, I(X;Y))$, and it satisfies $g_0(X,Y)=0$ and $g_{I(X;Y)}(X,Y)=\infty$.
    \item[3.] Let $D(Y)$ denote the "non-Gaussianness" of $Y$, defined as  $D(Y):=D(P_Y||P_{Y_\sG})$ (here $D(\cdot||\cdot)$ is the Kullback-Leibler divergence) with $Y_\sG$ being a Gaussian random variable having the same mean and variance as $Y$. Then we have $$\frac{1}{2}\log\left(1+\gamma_\eps2^{-2D(Y)}\var(Y)\right)\leq g_\eps(X,Y)\leq \frac{1}{2}\log(1+\gamma_\eps\var(Y)).$$
  \end{itemize}
\end{proposition}
\begin{proof}
Parts 1 and 2 can be proved directly from continuity and strict monotonicity of the maps $\gamma\mapsto I(Y; Z_\gamma)$ and $\gamma\mapsto I(X; Z_\gamma)$.
 The upper bound in part 3 is a direct consequence of the fact that a Gaussian input maximizes the mutual information between input and output of an additive Gaussian channel. The lower bound follows from the entropy power inequality \cite[Theorem 17.7.3]{Cover_Book} which states that
    $2^{2h(Z_\gamma)}\geq \gamma2^{2h(Y)}+2\pi e$ and hence
    $$g_\eps(X,Y)=I(Y; Z_{\gamma_\eps})\leq \frac{1}{2}\log\left(\gamma_\eps2^{2h(Y)}+2\pi e\right)-\frac{1}{2}\log(2\pi e),$$
from which and the fact that $D(Y)=h(Y_\sG)-h(Y)$, the lower bound immediately follows. \qed
\end{proof}
 In light of Proposition~\ref{Proposition:Properties}, it is clear that, unless $X$ and $Y$ are independent, $Z_\gamma$ is independent of $X$ if and only if $\gamma=0$, which implies $g_0(X,Y)=0$. As mentioned in the introduction, this is in contrast with the discrete rate-privacy function \eqref{gEpsilon}, where $g^{\mathsf{dis}}_0(X,Y)$ may be positive (for example, when $Y$ is an erased version of $X$, see \cite[Lemma 12]{Asoode_submitted}).
\begin{example}\label{Example_Gaussian}
Let $(X_\sG,Y_\sG)$ be a pair of Gaussian random variables with zero mean and correlation coefficient $\rho$. Then  $Z_\gamma$ is also a Gaussian random variable with variance $\gamma\var(Y_\sG)+1$. Without loss of generality assume that $Y_\sG$ has  unit variance. Then
$$I(X_\sG; Z_\gamma)=\frac{1}{2}\log\left(\frac{\gamma+1}{\gamma-\gamma\rho^2+1}\right),$$
and hence for any $\eps\in [0, I(X_\sG;Y_\sG))$ the equation $I(X_\sG;Z_\gamma)=\eps$  has the unique solution
$$\gamma_\eps=\frac{1-2^{-2\eps}}{2^{-2\eps}+\rho^2-1}.$$ Thus, we obtain \begin{eqnarray}
  g_\eps(X_\sG,Y_\sG)&=&\frac{1}{2}\log(1+\gamma_\eps)=\frac{1}{2}\log\left(\frac{\rho^2}{2^{-2\eps}+\rho^2-1}\right)\nonumber\\
  &=&\frac{1}{2}\log\left(1+\frac{2^{2\eps}-1}{1-2^{-2(I(X_\sG;Y_\sG)-\eps)}}\right).\label{Gaussian} \end{eqnarray}
The graph of $g_\eps(X_\sG,Y_\sG)$ is depicted in Fig.~\ref{Fig:   Gaussian} for $\rho=0.5$ and $\rho=0.8$. It is worth noting that $g_\eps(X_\sG,Y_\sG)$ is related to the Gaussian rate-distortion function $R_{\sG}(D)$ \cite{Cover_Book}. In fact,  $g_\eps(X_\sG, Y_\sG)=R_\sG(D_\eps)$ for $\eps\leq I(X_\sG; Y_\sG)$ where
$$D_\eps=\frac{2^{-2\eps}-2^{-2I(X_\sG; Y_\sG)}}{\rho^2},$$ is the mean squared distortion incurred in reconstructing $Y$ given the displayed data $Z_\gamma$.
\end{example}
\begin{figure}[t]
  \centering
  \includegraphics[width=10.5cm, height=8cm]{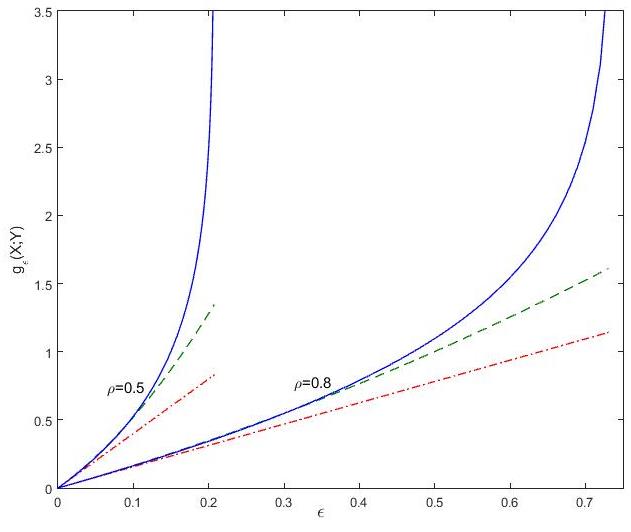}\\
  \caption{\small{The rate-privacy function for a pair of Gaussian $(X_\sG,Y_\sG)$, given by \eqref{Gaussian}, for $\rho=0.5$ and $\rho=0.8$. The first and second-order approximations are also shown in red and green, respectively. }}\label{Fig: Gaussian}
\end{figure}

The next result provides the first derivative $g'_\eps(X,Y)$ of the function $\eps\mapsto g_\eps(X,Y)$ at any $\eps<I(X;Y)$.
\begin{theorem}\label{Theorem:Derivative_G}
   For any absolutely continuous random variables $(X,Y)$, we have
   $$g'_{\eps}(X,Y)=\frac{\mmse(Y|Z_{\gamma_{\eps}})}{\mmse(Y|Z_{\gamma_{\eps}})-\mmse(Y|Z_{\gamma_{\eps}}, X)}.$$
    \end{theorem}
 \begin{proof}
   Since $g_{\eps}(X,Y)=I(Y; Z_{\gamma_\eps})$, we have
   \begin{eqnarray}
   \frac{\text{d}}{\text{d}\eps}g_{\eps}(X,Y)&=& \left[\frac{\text{d}}{\text{d}\gamma}I(Y; Z_{\gamma})\right]_{\gamma=\gamma_\eps}\frac{\text{d}}{\text{d}\eps}\gamma_\eps\nonumber\\
   &\stackrel{(a)}{=}&\frac{1}{2}\mmse(Y|Z_{\gamma_\eps})\frac{\text{d}}{\text{d}\eps}\gamma_\eps,\label{derivative}
   \end{eqnarray}
   where $(a)$ follows from \eqref{I_MMSE}. In order to calculate $\frac{\text{d}}{\text{d}\eps}\gamma_\eps$, notice that $\eps=I(X; Z_{\gamma_\eps})$ and hence taking the derivative of both sides of this equation with respect to $\eps$ yields
   $$1=\left[\frac{\text{d}}{\text{d}\gamma}I(X; Z_{\gamma})\right]_{\gamma=\gamma_\eps}\frac{\text{d}}{\text{d}\eps}\gamma_\eps,$$ and hence
   \begin{eqnarray}
     \frac{\text{d}}{\text{d}\eps}\gamma_\eps &=& \frac{1}{ \left[\frac{\text{d}}{\text{d}\gamma}I(X; Z_{\gamma})\right]_{\gamma=\gamma_\eps}}  \nonumber\\
      &\stackrel{(a)}{=}& \frac{2}{\mmse(Y|Z_{\gamma_\eps})-\mmse(Y|Z_{\gamma_\eps}, X)}, \label{derivative_Gamma}
   \end{eqnarray}
where $(a)$ follows from \eqref{I_MMSE2}. The result then follows by plugging \eqref{derivative_Gamma} into \eqref{derivative}. \qed

 \end{proof}
As a simple illustration of Theorem~\ref{Theorem:Derivative_G}, consider jointly Gaussian $X_\sG$ and $Y_\sG$ whose rate-privacy function is computed in Example~\ref{Example_Gaussian}. In particular, \eqref{Gaussian} gives
\begin{equation}\label{Derivative_Gaussian}
  g'_\eps(X_\sG,Y_\sG)=\frac{2^{-2\eps}}{2^{-2\eps}+\rho^2-1}.
\end{equation}
On the other hand, since $X_\sG=\sqrt{\alpha}Y_\sG+N_1$ where $\alpha=\rho^2\var(X)$, $N_1\sim \N(0, \sigma_N^2)$ is independent of $Y_\sG$, and $\sigma_N^2=(1-\rho^2)\var(X)$, one can conclude from \cite[Proposition 3]{MMSE_Guo} that $$\mmse(Y_\sG|Z_\gamma, X_\sG)=\mmse\left(Y_\sG|Z_\gamma, \frac{1}{\sigma^2_N}X_\sG\right)=\mmse(Y_\sG|Z_{\gamma+a}),$$ where $a=\frac{\rho^2}{1-\rho^2}$. Recalling that $\mmse(Y_\sG|Z_\gamma)=\frac{1}{1+\gamma}$, we  obtain
\begin{eqnarray*}
  \frac{\mmse(Y_\sG|Z_\gamma)}{\mmse(Y_\sG|Z_\gamma)-\mmse(Y_\sG|Z_{\gamma+a})}&=&\frac{1+(1-\rho^2)\gamma_\eps}{\rho^2}\\ &=&\frac{2^{-2\eps}}{2^{-2\eps}+\rho^2-1},
  \end{eqnarray*}
  which equals \eqref{Derivative_Gaussian}.

In light of Theorem~\ref{Theorem:Derivative_G}, we can now show that the map $\eps\mapsto g_\eps(X,Y)$ is in fact infinitely differentiable over $(0, I(X; Y))$.
\begin{corollary}\label{corollary_Smoothness}
  For a pair of absolutely continuous $(X,Y)$, the map $\eps\mapsto g_\eps(X,Y)$ is infinitely differentiable at any $\eps\in(0, I(X;Y))$. Moreover, if all the moments of $Y$ is finite, then $\eps\mapsto g_\eps(X,Y)$ is infinitely right differentiable at $\eps=0$.
\end{corollary}
\begin{proof}
  It is shown in \cite[Proposition 7]{MMSE_Guo_Wu} that $\gamma\mapsto \mmse(Y|Z_\gamma)$ is infinitely differentiable at any $\gamma>0$ and  infinitely right differentiable at $\gamma=0$ if all the moments of $Y$ are finite. Thus the corollary follows from Theorem~\ref{Theorem:Derivative_G} noting that since $\E[Y^k]\fin$ for all $k$, we also have $\E[Y^k|X=x]\fin$ for almost all $x$ (except for $x$ in a set of zero $P_X$-measure). It therefore follows that $\gamma\mapsto \mmse(Y|Z_\gamma, X)$ is also infinitely right differentiable at $\gamma=0$.\qed
\end{proof}
We remark that using \eqref{Second_Derivative_MMSE} and Theorem~\ref{Theorem:Derivative_G}, one can easily calculate the second derivative as
\begin{eqnarray}
  g''_{\eps}(X,Y) &=& \frac{\text{d}^2}{\text{d}\eps^2}g_\eps(X,Y)\label{Second_Derivative_g_EPS}\\
   &=&\frac{2\left(\mmse(Y|Z_{\gamma_\eps}, X)\E[\var^2(Y|Z_{\gamma_\eps})]-\mmse(Y|Z_{\gamma_\eps})\E[\var^2(Y|Z_{\gamma_\eps}, X)]\right)}{\left[\mmse(Y|Z_{\gamma_{\eps}})-\mmse(Y|Z_{\gamma_\eps}, X)\right]^3}.\nonumber
\end{eqnarray}

The following corollary, which is an immediate consequence of Theorem~\ref{Theorem:Derivative_G}, provides a second-order approximation for $g_\eps(X,Y)$ as $\eps\downarrow 0$ and thus an approximation to the the rate-privacy function in the almost perfect privacy regime.
\begin{corollary}\label{corollary_Approximation_g_eps}
For a given pair of absolutely continuous random variables $(X,Y)$, we have as $\eps\downarrow 0$,
$$g_\eps(X,Y)=\frac{\eps}{\eta^2_X(Y)}+\Delta(X,Y)\eps^2+o(\eps^2),$$
where
\begin{equation}\label{Second_Derivative_Delta}
  \Delta(X,Y)=\frac{1}{\eta_X^4(Y)}\left(\frac{\var^2(Y)-\E[\var^2(Y|X)]}{\var^2(Y)\eta_X^2(Y)}-1\right),
\end{equation} and $\eta^2_X(Y)$ is the one-sided maximal correlation between $X$ and $Y$ defined in \eqref{Eta_Definition}.
\end{corollary}
\begin{proof}
  According to Corollary~\ref{corollary_Smoothness}, we can use the second-order Taylor expansion to approximate $g_\eps(X,Y)$ around $\eps=0$, resulting in
  $$g_\eps(X,Y)=\eps g'_0(X,Y)+\frac{\eps^2}{2}g''_0(X,Y)+o(\eps^2).$$
From Theorem~\ref{Theorem:Derivative_G} and \eqref{Second_Derivative_g_EPS} we have $g'_0(X,Y)=\frac{1}{\eta^2_X(Y)}$ and $g''_0(X,Y)=2\Delta(X,Y)$, respectively, from which the corollary follows. \qed
  \end{proof}
 Since $\rho_m^2(X_\sG, Y_\sG)=\rho^2$ for jointly Gaussian $X_\sG$ and  $Y_\sG$ with correlation coefficient $\rho$, \eqref{Eta_Gaussian_inequality} implies that $\eta^2_{X_\sG}(Y_\sG)=\rho^2$ and  $\Delta(X_\sG,Y_\sG)=\frac{1-\rho^2}{\rho^4}$, and therefore Corollary~\ref{corollary_Approximation_g_eps} implies that
for small $\eps> 0$,
$$g_\eps(X_\sG,Y_\sG)=\frac{1}{\rho^2}\eps+\frac{1-\rho^2}{\rho^4}\eps^2+o(\eps^2).$$
This second-order approximation as well as the  first-order approximation are illustrated in Fig.~\ref{Fig:   Gaussian} for $\rho=0.5$ and $\rho=0.8$.

 Polyanskiy and Wu \cite{Dissipation_Polyanskiey} have recently generalized the strong data processing inequality of Anantharam et al.\ \cite{anantharam} for the case of continuous random variables $X$ and $Y$ with joint distribution $P_{XY}$. Their result states that  \begin{equation}\label{strong_Data_Anantharam}
  \sup_{X\markov Y\markov U,\atop  0<I(U;Y)\fin }\frac{I(X;U)}{I(Y;U)}=S^*(Y,X),
\end{equation}
where $$S^*(Y,X):=\sup_{Q_Y, \atop 0<D(Q_Y||P_Y)\fin} \frac{D(Q_X||P_X)}{D(Q_Y||P_Y)},$$
 where $P_X$ and $P_Y$ are the marginals of $P_{XY}$ and  $Q_X(\cdot)=\int P_{X|Y}(\cdot|y)Q_Y(\text{d}y)$. In addition, it is shown in \cite{Dissipation_Polyanskiey} that the supremum in \eqref{strong_Data_Anantharam} is achieved by a binary $U$.
 Replacing $U$ with $Z_\gamma$, we can conclude from \eqref{strong_Data_Anantharam} that
 $$\frac{I(X; Z_\gamma)}{I(Y; Z_\gamma)}\leq S^*(Y,X),$$ for any $\gamma\geq 0$. Letting  $\gamma=\gamma_\eps$, the above yields that
 \begin{equation}\label{LowerBound_S*}
    g_\eps(X,Y)\geq \frac{\eps}{S^*(Y,X)}.
 \end{equation}
Clearly, this bound may be expected to be tight only for small $\eps>0$ since $g_\eps(X,Y)\to \infty$ as $\eps\to I(X;Y)$, as shown in Proposition~\ref{Proposition:Properties}. Note that Theorem~\ref{Theorem:Derivative_G} implies $\lim_{\eps\downarrow 0}\frac{g_\eps(X,Y)}{\eps}=\frac{1}{\eta_X^2(Y)}$. On the other hand, it can be easily shown that $\eta_X^2(Y) \leq S^*(Y,X)$, with equality when $X$ and $Y$ are jointly Gaussian and hence the inequality \eqref{LowerBound_S*} becomes tight for small $\eps$ and jointly Gaussian $X$ and $Y$.

The bound in \eqref{LowerBound_S*} would be significantly improved if we could show that $g_\eps(X,Y)\geq g_\eps(X_\sG, Y_\sG)$, where $X_\sG$ and $Y_\sG$ are jointly Gaussian having the same means, variances, and correlation coefficient as $(X, Y)$. This is because in that case we could write
\begin{equation}\label{Lower_Bound_WRONG}
  g_\eps(X, Y)\geq g_\eps(X_\sG, Y_\sG)\stackrel{(a)}{\geq} \frac{\eps}{\eta^2_{X_\sG}(Y_\sG)}=\frac{\eps}{\rho^2(X_\sG,Y_\sG)}=\frac{\eps}{\rho^2(X,Y)}\stackrel{(b)}{\geq} \frac{\eps}{\eta_X^2(Y)},
\end{equation}  where $(a)$ and $(b)$ follow from \eqref{Gaussian} and  \eqref{Eta_Gaussian_inequality}, respectively. However, as shown in Appendix~\ref{Appendix A}, the inequality $g_\eps(X,Y)\geq g_\eps(X_\sG, Y_\sG)$ does not in general hold\footnote{We will see in the next section that this holds in the estimation-theoretic formulation of privacy, i.e., the Gaussian case is the \emph{worst} case when the privacy filter is an additive Gaussian channel and the utility and privacy are measured as $\mmse(Y|Z_\gamma)$ and $\mmse(X|Z_\gamma)$, respectively.}.
It is therefore possible to have $g_\eps(X,Y)<\frac{\eps}{\eta_X^2(Y)}$ for some $0<\eps< I(X;Y)$.  To construct an example, it suffices to construct $P_{XY}$ for which $\eps\mapsto g_\eps(X,Y)$ is locally concave at zero (i.e., $g''_0(X,Y)<0$) and hence its graph lies below the tangent line $\frac{\eps}{\eta_X^2(Y)}$ for some $\eps>0$. Let $Y\sim \N(0, 1)$ and $X=Y\cdot1_{\{Y\in[-1, 1]\}}$. Then it can be readily shown that $\E[\var(Y|X)]<\E[\var^2(Y|X)]$, which implies that $\Delta(X,Y)<0$. Hence, since $g''_0(X,Y)=2\Delta(X,Y)$, we have that $g''(X,Y)<0$. This observation is illustrated in Fig.~\ref{Fig: Truncated_Gaussian}.
\begin{figure}[t]
  \centering
  \includegraphics[width=10.5cm, height=8cm]{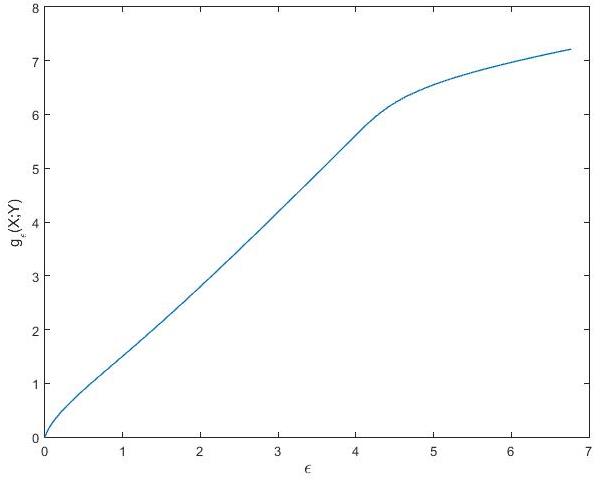}\\
  \caption{\small{The rate-privacy function for $Y\sim \N(0,1)$ and $X=Y\cdot1_{\{Y\in[-1, 1]\}}$. The map $\eps\mapsto g_\eps(X,Y)$ is clearly locally concave at zero. Note that here $I(X;Y)=\infty$ and hence $\eps$ is unbounded.}}\label{Fig: Truncated_Gaussian}
\end{figure}


As remarked earlier, the map $\eps\mapsto g_\eps(X,Y)$ is in general not convex and thus one cannot conclude that $g'_\eps(X,Y)\geq g'_0(X,Y)=\frac{1}{\eta_X^2(Y)}$. However, it can be shown that this implication holds if $P_{XY}$ has more structure. In the next theorem, we assume that $Y$ is a noisy version of $X$ through an additive Gaussian channel.
\begin{theorem}
  For a given $X\sim P_X$ with variance $\sigma_X^2$, and $Y=aX+M_\sG$ with $M_\sG\sim \N(0, \sigma^2_M)$ independent of $X$, we have:
  \begin{itemize}
    \item[1.] If $a^2\sigma_X^2\geq \sigma_M^2$, then $\eps\mapsto g_\eps(X,Y)$ is convex.
    \item[2.] For any $a>0$ and $\eps\in[0, I(X;Y))$, we have
        \begin{equation}\label{Theorem_1}
          g_\eps(X,Y)\geq \frac{\eps}{\eta_X^2(Y)}.
        \end{equation}
        Furthermore, we have
        \begin{equation}\label{Theorem_2}
          \inf_{\gamma\geq 0}\frac{\mmse(Y|Z_\gamma, X)}{\mmse(Y|Z_\gamma)}=1-\eta_X^2(Y),
        \end{equation}
        and
        \begin{equation}\label{Theorem_3}
          \sup_{\gamma> 0}\frac{I(X;Z_\gamma)}{I(Y;Z_\gamma)}=\eta_X^2(Y).
        \end{equation}
    \end{itemize}
\end{theorem}
\begin{proof}
The first part follows from a straightforward computation showing that  if $a^2\var(X)\geq \sigma^2_M$, then $\Delta(X,Y)\geq 0$.

To prove the second part, note that for any $\gamma\geq0$ we have
\begin{eqnarray}
  \mmse(Y|Z_\gamma) &=& \mmse(aX+M_\sG|a\sqrt{\gamma}X+\sqrt{\gamma}M_\sG+N_\sG)\nonumber\\
   &\stackrel{(a)}{=}& \frac{1}{\gamma}\mmse\left(N_\sG|a\sqrt{\gamma}X+\sqrt{\gamma}M_\sG+N_\sG\right)\nonumber\\
   &\stackrel{(b)}{\leq}& \frac{a^2\var(X)+\sigma_M^2}{1+\gamma(a^2\var(X)+\sigma_M^2)}< \frac{a^2\var(X)+\sigma_M^2}{1+\gamma\sigma_M^2}\nonumber\\
   &\stackrel{(c)}{=}&\frac{1}{\gamma}\left(\frac{a^2\var(X)+\sigma_M^2}{\sigma_M^2}\right)\mmse\left(N_\sG|\sqrt{\gamma}M_\sG+N_\sG\right)\nonumber\\
   &\stackrel{(d)}{=}&\left(\frac{a^2\var(X)+\sigma_M^2}{\sigma_M^2}\right)\mmse(Y|Z_\gamma, X),\label{Additive_Derivative_proof}
\end{eqnarray}
where $(a)$ follows from the fact that $\mmse(U|\alpha U+V)=\frac{1}{\alpha^2}\mmse(V|\alpha U+V)$ for $\alpha\neq 0$, $(b)$ and $(c)$ hold by \cite[Theorem 12]{MMSE_WU_Properties} which states that $\mmse(U|U+V_\sG)\leq \mmse(U_\sG|U_\sG+V_\sG)=\frac{\var(U)\var(V)}{\var(U)+\var(V)}$. Finally, $(d)$ follows from the following chain of equalities
\begin{eqnarray*}
  \mmse(Y|Z_\gamma, X) &=& \mmse(aX+M_\sG|a\sqrt{\gamma}X+\sqrt{\gamma}M_\sG+N_\sG, X) \\
   &=& \mmse(M_\sG|\sqrt{\gamma}M_\sG+N_\sG, X)\\
   &\stackrel{(e)}{=}& \mmse(M_\sG|\sqrt{\gamma}M_\sG+N_\sG)\\
   &=&\frac{1}{\gamma}\mmse(N_\sG|\sqrt{\gamma}M_\sG+N_\sG)
\end{eqnarray*}
where $(e)$ holds since $X$ and $M_\sG$ are independent.

We can therefore write
\begin{eqnarray}
  g'_\eps(X,Y) &=& \frac{\mmse(Y|Z_{\gamma_\eps})}{\mmse(Y|Z_{\gamma_\eps})-\mmse(Y|Z_{\gamma_\eps}, X)}\nonumber\\
  &\stackrel{(a)}{\geq} & \frac{a^2\var(X)+\sigma_M^2}{a^2\var(X)}\stackrel{(b)}{=}\frac{1}{\eta_X^2(Y)}=g'_0(X,Y), \label{Additive_Derivative}
\end{eqnarray}
where $(a)$ is due to \eqref{Additive_Derivative_proof} and $(b)$ holds since $\var(Y)=a^2\var(X)+\sigma_M^2$ and $\var(\E[Y|X])=a^2\var(X)$. The identity $g_\eps(X,Y)=\int_0^\eps g'_t(X,Y)\text{d}t$, and inequality \eqref{Additive_Derivative} together imply that $g_\eps(X,Y)\geq \frac {\eps}{\eta_X^2(Y)}$ for $\eps\leq I(X;Y)$.

Furthermore, according to Theorem~\ref{Theorem:Derivative_G}, the inequality \eqref{Additive_Derivative} yields \eqref{Theorem_2}.
 Using the integral representation of mutual information in \eqref{I_MMSE} and \eqref{I_MMSE2},  we can write for any $\gamma\geq 0$
\begin{eqnarray}
  I(X; Z_\gamma) &=& \frac{1}{2}\int_0^\gamma\left[\mmse(Y|Z_t)-\mmse(Y|Z_t, X)\right]\text{d}t \nonumber\\
   &\leq & \frac{\eta_X^2(Y)}{2}\int_0^\gamma\mmse(Y|Z_t)\text{d}t=\eta_X^2(Y)I(Y; Z_\gamma),\label{Proof_Theorem_3}
\end{eqnarray} where the inequality is due to \eqref{Theorem_2}. The equality \eqref{Theorem_3} then follows from \eqref{Proof_Theorem_3}.
\qed
\end{proof}

It should be noted that both MMSE and mutual information satisfy the data processing inequality, see, \cite{MMSE_WU_Properties} and \cite{anantharam}, that is, $\mmse(U|V)\leq \mmse(U|W)$, and $I(U; W)\leq I(U; V)$ for $U\markov V\markov W$. Therefore, \eqref{Theorem_2} can be thought of as a strong version of the data processing inequality for MMSE for the trivial Markov chain $Y\markov (Z_\gamma, X)\markov Z_\gamma$. Also, \eqref{Theorem_3} can be viewed as a strong data processing inequality for the mutual information for the Markov chain $X\markov Y\markov Z_\gamma$ which is slightly stronger than \eqref{strong_Data_Anantharam} in  the special case of an additive Gaussian channel as $\eta_X^2(Y)\leq S^*(Y, X)$.

\section{Estimation-Theoretic Formulation}
Consider the same scenario as in the previous section: Alice observes $Y$, which is correlated with the private data $X$ according to a given joint distribution $P_{XY}$, and wishes to transmit a random variable $Z$ to Bob to receive a utility from him.  An \emph{operational} measure of privacy is proposed in \cite{Asoode_MMSE_submitted} where Alice generates the displayed data $Z$ via a privacy filter $P_{Z|Y}$ such that Bob cannot efficiently estimate any non-trivial function of $X$ given $Z$. As before, her goal is to maximize the utility (or equivalently minimize the cost) between $Y$ and the displayed data $Z$. The next definition formalizes this privacy guarantee. We call a function $f$ of random variable $X$ \emph{non-degenerate} if $f(X)$ is not almost everywhere constant with respect to the probability measure $P_X$. Also, we assume throughout this section that $X$ and $Y$ have finite second moments.
\begin{definition}\label{Def:strong_estimation_privacy}
Given a pair of jointly absolutely continuous random variables $(X,Y)$ with joint distribution $P_{XY}$ and $0\leq\eps\leq 1$, we say $Z$ satisfies \emph{$\eps$-strong estimation privacy}, if there exists a channel $P_{Z|Y}$ that induces a joint distribution $P_X\times P_{Z|X}$, via the Markov condition $X\markov Y\markov Z$, satisfying
\begin{equation}\label{Eq:strong_estimation_privacy}
  \mmse(f(X)|Z)\geq (1-\eps)\var(f(X)), ~~~\text{or equivalently,}~~~\eta^2_Z(f(X))\leq \eps,
\end{equation}
for any non-degenerate Borel function $f$. Similarly, $Z$ is said to satisfy  \emph{$\eps$-weak estimation privacy}, if \eqref{Eq:strong_estimation_privacy} is satisfied only for the identity function $f(x)=x$.
\end{definition}
It is shown in \cite{Asoode_MMSE_submitted} that $\eps$-strong estimation privacy is equivalently characterized by the requirement $\rho_m^2(X, Z)\leq \eps$. In other words, $\mmse(f(X)|Z)\geq (1-\eps)\var(f(X))$ for any non-degenerate Borel function $f$ if and only if $\rho_m^2(X,Z)\leq \eps$. Let the utility that Alice receives from Bob be measured by $\frac{\var(Y)}{\mmse(Y|Z)}$, which she aims to maximize. For mathematical convenience, we define the \emph{cost} that Alice suffers by describing $Z$ in lieu of $Y$ as the estimation noise-to-signal ratio (ENSR), $\frac{\mmse(Y|Z)}{\var(Y)}$, and hence Alice equivalently aims to minimize the ENSR. Focusing on additive Gaussian privacy filter $Z=Z_\gamma$, we can formalize the privacy-utility tradeoff as
$$\sM_{\eps}(X,Y):=\inf_{\gamma\in \C_{\eps}(P_{XY})} \frac{\mmse(Y|Z_\gamma)}{\var(Y)}=   1-\sup_{\gamma\in \C_{\eps}(P_{XY})}\eta^2_{Z_{\gamma}}(Y),$$
where $\C_{\eps}(P_{XY})$ is the set of parameters $\gamma$ corresponding to $\eps$-strong privacy, i.e.,
$$\C_{\eps}(P_{XY}):=\{\gamma\geq 0:\rho_m^2(X, Z_{\gamma})\leq \eps\}.$$
Similarly,
$$\sW_{\eps}(X,Y):=1-\sup_{\gamma\in \partial\C_{\eps}(P_{XY})}\eta^2_{Z_{\gamma}}(Y),$$
where $$\partial\C_{\eps}(P_{XY}):=\{\gamma\geq 0:\eta^2_{Z_{\gamma}}(X)\leq\eps\}.$$

Note that both the maximal correlation and the one-sided maximal correlation satisfy the data processing inequality, that is, $\rho_m^2(X,Z_\gamma)\leq \rho_m^2(Y, Z_\gamma)$ and $\eta^2_{Z_\gamma}(X)\leq \eta_{Y}(X)$. Therefore, in the definition of $\sM_\eps(X,Y)$ and $\sW_\eps(X,Y)$, we can restrict $\eps$ as $0\leq\eps\leq \rho_m^2(X,Y)$ and $0\leq\eps\leq \eta_Y^2(X)$, respectively.
\begin{example}\label{Example2}
Let $X_{\sG}$ and $Y_{\sG}$ be jointly Gaussian with correlation coefficient $\rho$. Without loss of generality assume  that $\E[X_{\sG}]=\E[Y_{\sG}]=0$. Since  $\rho^2_m(X_{\sG},Z_{\gamma})=\rho^2(X_{\sG}, Z_{\gamma})$, we have $$\rho_m^2(X_{\sG}, Z_{\gamma})=\rho^2\frac{\gamma\var(Y_\sG)}{1+\gamma\var(Y_\sG)},$$ which implies that the mapping  $\gamma\mapsto \rho_m^2(X_{\sG}, Z_{\gamma})$ is strictly increasing. Also, the equation $\rho_m^2(X_{\sG}, Z_{\gamma})=\eps$ for $0\leq\eps\leq\rho^2_m(X_{\sG},Y_{\sG})=\rho^2$ has a unique solution
$$\gamma_{\eps}:=\frac{\eps}{\var(Y_\sG)(\rho^2-\eps)},$$ and $\rho_m^2(X, Z_\gamma)\leq \eps$ for any $\gamma\leq \gamma_{\eps}$. On the other hand, $\mmse(Y_{\sG}|Z_{\gamma})=\frac{\var(Y_\sG)}{1+\gamma\var(Y_\sG)}$, which shows that the map $\gamma\mapsto \mmse(Y_{\sG}|Z_{\gamma})$ is strictly decreasing. Hence,
\begin{equation}\label{M_eps_gaussian}
  \sM_{\eps}(X_{\sG},Y_{\sG})=\frac{\mmse(Y_{\sG}|Z_{\gamma_{\eps}})}{\var(Y_{\sG})}=1-\frac{\eps}{\rho^2}.
\end{equation}
Clearly for jointly Gaussian $X_\sG$ and  $Y_\sG$ we have $\eta^2_{Z_{\gamma}}(X_{\sG})=\rho_m^2(X_{\sG},Z_{\gamma})=\eps$, for any $\gamma\geq 0$
and consequently $\C_\eps(P_{X_\sG Y_\sG})=\partial\C_\eps(P_{X_\sG Y_\sG})$, that is, for $0\leq \eps\leq \rho^2$,
\begin{equation}\label{Eq:SM_SW_Gaussian}
  \sM_{\eps}(X_{\sG},Y_{\sG})=\sW_{\eps}(X_{\sG},Y_{\sG})=1-\frac{\eps}{\rho^2}.
\end{equation}
\end{example}

Unlike $g_\eps(X, Y)$, the quantity $\sM_\eps(X,Y)$ is maximized among all pairs of random variables $(X,Y)$ with identical means, variances and correlation coefficient when $X$ and $Y$ are jointly Gaussian. Thus, Example~\ref{Example2} yields a sharp upper-bound for $\sM_\eps(X,Y)$. This is stated in the following theorem.
\begin{theorem}[\cite{Asoode_MMSE_submitted}]\label{Theorem_mmse_UpperBound}
For any given jointly absolutely continuous $(X,Y)$, we have for $0\leq \eps\leq \rho_m^2(X,Y)$, $$\sW_{\eps}(X,Y)\leq\sM_{\eps}(X,Y)\leq \sM_{\eps}(X_\sG,Y_\sG)= 1-\frac{\eps}{\rho_m^2(X,Y)},$$ where $(X_\sG, Y_\sG)$ is a pair of Gaussian random variables with the same means, variances, and correlation coefficient as $(X,Y)$.
\end{theorem}

Next, we turn our attention to the approximation of $\sM_\eps(X,Y)$ in the almost perfect privacy regime. Unfortunately, there is no known approximation for $\rho_m^2(X, Z_\gamma)$ and $\mmse(X|Z_\gamma)$ around $\gamma=0$. Nevertheless, we can use the first-order approximation of $g_\eps(X,Y)$ to derive an approximation for $\sM_\eps(X,Y)$ around $\eps=0$. The next theorem shows this approximation for the special case where $P_{Y|X}$ is an additive noise channel.
\begin{theorem}\label{theorem_G_ENSR_ALMOST}
If $X\sim \N(b, \sigma_X^2)$ and $Y=aX+M$, where $a, b\in \R^+$, and $M$ is a noise random variable having a density, then for sufficiently small $\eps$
\begin{equation}\label{MMSE_Info_privacy_Final}
\sM_\eps(X_\sG,Y)\geq 2^{-D(Y)}2^{-2g_{\eps+o(\eps)}(X_\sG, Y)}.
\end{equation}
\end{theorem}

\begin{proof}

We start by deriving an inequality relating $\mmse(Y|Z_\gamma)$ and $I(Y; Z_\gamma)$ which originates from the Shannon lower bound for the rate-distortion function. Since the Gaussian distribution maximizes the differential entropy \cite[Theorem 8.6.5]{Cover_Book}, we have $h(Y|Z=z)\leq \frac{1}{2}\log(2\pi e\var(Y|Z=z))$ for any random variable $Z$. It immediately follows from Jensen's inequality that $$h(Y|Z_\gamma)\leq \frac{1}{2}\log(2\pi e\mmse(Y|Z_\gamma)),$$ and hence
\begin{equation}\label{MMSE_Info2_1_Privacy}
\mmse(Y|Z_\gamma)\geq  \frac{1}{2\pi e}2^{2h(Y|Z_\gamma)}=\var(Y)2^{2(h(Y)-h(Y_\sG))}2^{-2I(Y;Z_\gamma)},
\end{equation}
from which we obtain
\begin{equation}\label{MMSE_Info2_Privacy}
  \inf_{\begin{smallmatrix}\gamma\geq0,\\I(X;Z_\gamma)\leq\eps\end{smallmatrix}}\frac{\mmse(Y|Z_\gamma)}{\var(Y)}\geq 2^{-D(Y)}2^{-2g_{\eps}(X,Y)},
\end{equation}
 where $D(Y)$ is the non-Gaussianness of $Y$ defined in Proposition~\ref{Proposition:Properties}. We note that a similar inequality is proved in \cite[Lemma 13]{Asoode_submitted} for arbitrary noise distribution provided that $Y$ is Gaussian.
 Although, inequality \eqref{MMSE_Info2_Privacy} provides an operational interpretation of $g_{\eps}(X,Y)$, it does not relate  $g_\eps(X,Y)$ to $\sM_\eps(X,Y)$. Such a relationship would follow if
 $\rho_m^2(X,Z_\gamma)\leq \eps$ implied $I(X; Z_\gamma)\leq \eps$ for a given $(X,Y)$, because then according to \eqref{MMSE_Info2_Privacy}, one could conclude that $\sM_\eps\geq 2^{-D(Y)}2^{-2g_\eps(X,Y)}$.
 However, this implication does not hold in general. Nevertheless, we show in the sequel that this implication holds for Gaussian $X$ in the almost perfect privacy regime when $P_{Y|X}$ is an additive noise channel. First we notice that for jointly Gaussian $X_\sG$ and $Y_\sG$, we have  $I(X_\sG; \sqrt{\gamma}Y_\sG+N_\sG)=-\frac{1}{2}\log(1-\rho^2(X_\sG, \sqrt{\gamma}Y_\sG+N_\sG))$. Hence, since $\rho_m^2(X_\sG,\sqrt{\gamma}Y_\sG+N_\sG)=\rho^2(X_\sG, \sqrt{\gamma}Y_\sG+N_\sG)$, the above implication clearly holds, i.e., $\rho_m^2(X_\sG,\sqrt{\gamma}Y_\sG+N_\sG)\leq \eps$ implies $I(X_\sG; \sqrt{\gamma}Y_\sG+N_\sG)\leq \eps$. On the other hand, specializing the decomposition~\eqref{decompositon_MI} proved in Appendix~\ref{Appendix A} for $U=X_\sG$ and $V=Z_\gamma$, we can write \begin{equation}\label{Non_gaussianiness1}
   I(X_\sG; Z_\gamma)=I(X_{\sG}; \sqrt{\gamma}Y_{\sG}+N_{\sG})+D(Z_{\gamma}|X_\sG)-D(Z_{\gamma}),
 \end{equation}
where $D(V|U)$ for a pair of absolutely continuous random variables $(U,V)$ is defined as
\begin{equation}\label{Non-Gaussianness_Conditional}
  D(V|U):=D(P_{V|U}||P_{V_\sG |U_\sG}|P_U)=\E_{UV}\left[\log\frac{P_{V|U}}{P_{V_\sG |U_\sG}}\right],
\end{equation}
where $(U_\sG, V_\sG)$ is a pair of Gaussian random variables having the same means, variances and correlation coefficient as $(U, V)$, and $P_{V_\sG|U_\sG}(\cdot|u)$ and $P_{V|U}(\cdot|u)$ are the conditional densities of $V_\sG$ and $V$ given $U_\sG=u$ and $U=u$, respectively.
 As shown in \cite[Appendix II]{MMSE_Guo}  if $\var(Y)\fin$, then as $\gamma\to 0$
 \begin{equation}\label{Non-Gaussiness_Small_SNR}
   D(Z_\gamma)=o(\gamma).
 \end{equation}
  Lemma~\ref{Lemma_Vanishing_Conditional_NONG} in Appendix~\ref{Appendix_Proof_Lemma}
  shows that $D(Z_\gamma|X_\sG)$ also behaves like $o(\gamma)$ if $\mmse(Y|X_\sG)=\mmse(Y_\sG|X_\sG)$.
  In light of this lemma, \eqref{Non_gaussianiness1}, and \eqref{Non-Gaussiness_Small_SNR}, we can conclude that
$$I(X_\sG; Z_\gamma)\leq I(X_\sG; \sqrt{\gamma}Y_\sG+N_\sG)+\frac{\gamma}{2}\left[\mmse(Y_\sG|X_\sG)-\mmse(Y|X_\sG)\right]+o(\gamma).$$
 Thus if $P_{XY}$ satisfies $\mmse(Y|X_\sG)=\mmse(Y_\sG|X_\sG)$, or equivalently $\E[\var(Y|X_\sG)]=1-\rho^2(X,Y)$, we have
 \begin{equation}\label{Gaussian_ENSR_G11}
   I(X_\sG; Z_\gamma)\leq I(X_\sG; \sqrt{\gamma}Y_\sG+N_\sG)+o(\gamma).
 \end{equation}
Since $\rho_m^2(X_\sG, Z_\gamma)\geq \rho_m^2(X_\sG, \sqrt{\gamma}Y_\sG+N_\sG)$, we can conclude from \eqref{Gaussian_ENSR_G11} that,  $\rho_m^2(X_\sG, Z_\gamma)\leq \eps$ implies $I(X_\sG; Z_\gamma)\leq \eps+o(\gamma)$ for sufficiently small $\gamma$ (or equivalently $\eps$). Note that it is straightforward to show that $\rho_m^2(X_\sG, Z_\gamma)\leq \eps$ implies $\gamma\leq \frac{\eps}{\rho^2(X_\sG, Y)-\eps}$ (see Example~\ref{Example2}). Hence, in the almost perfect privacy regime, $\rho_m^2(X_\sG, Z_\gamma)\leq \eps$ is satisfied with $\gamma$ which is at most linear in $\eps$. Therefore, \eqref{Gaussian_ENSR_G11} allows us to conclude that $\rho_m^2(X_\sG, Z_\gamma)\leq \eps$ implies that $I(X_\sG; Z_\gamma)\leq \eps+o(\eps)$.

The condition $\E[\var(Y|X_\sG)]=1-\rho^2(X,Y)$ is satisfied if the channel from $X_\sG$ to $Y$ is additive, that is, $Y=aX_\sG+M$, where $a\in \R^+$ and $M$ is a noise random variable with a density having variance $1-\rho^2(X_\sG,Y)$. However, since $\E[\var(Y|X_\sG)]=\E[\var(Y|rX_\sG)]$ for any  $r\neq 0$, the variance condition can be removed.\qed
\end{proof}

The lower-bound \eqref{MMSE_Info_privacy_Final} can be further simplified by invoking Corollary~\ref{corollary_Approximation_g_eps}, which results in
$$\sM_\eps(X_\sG,Y)\geq  2^{-D(Y)}\left(1-\frac{2\eps}{ \eta_{X_\sG}^2(Y)}\right)+o(\eps).$$


One the other hand, as proved in \cite{Asoode_MMSE_submitted}, when $Y$ is Gaussian, $Y_\sG$, then $$1-\frac{\eps}{\rho^2(X, Y_\sG)}\leq \sM_\eps(X, Y_\sG)\leq 1-\frac{\eps}{\rho^2_m(X,Y_\sG)},$$ for any $\eps\leq \rho^2_m(X,Y)$. We have therefore tight lower bounds for $\sM_\eps(X,Y)$ when either $X$ or $Y$ is Gaussian.

\section{Conclusion}
In this paper, we studied the problem of approximating the maximal amount of information  one can transmit about a random variable $Y$ over an additive Gaussian channel without revealing more than a certain (small) amount of information about another random variable $X$ that represents sensitive or private data. Specifically, letting $g_\eps(X,Y)$ denote the maximum of  $I(Y;Z_\gamma)$ over $\gamma\geq 0$, where $Z_\gamma:=\sqrt{\gamma}Y+N_\sG$ and $N_\sG\sim \N(0,1)$ is independent of $(X,Y)$, subject to $I(X; Z_\gamma)\leq \eps$, we showed that $g_\eps(X,Y)=\frac{\eps}{\eta_X^2(Y)}+\Delta(X,Y)\eps^2+o(\eps)$ where $\eta_X^2(Y)$ and $\Delta(X,Y)$ are two asymmetric measures of correlation between $X$ and $Y$. For the special case of jointly Gaussian $X$ and $Y$, the approximation was compared with  the exact value of $g_\eps(X,Y)$. As a side result, we also showed that this approximation leads to a slightly improved version of the strong data processing inequality under some suitable conditions on $P_{Y|X}$.

We also studied an estimation-theoretic formulation of the privacy-utility tradeoff for the same setup. Let $\sM_\eps(X,Y)$ be the smallest achievable MMSE in estimating $Y$ given $Z_\gamma$ such that MMSE in estimating any function $f$ of $X$ given $Z_\gamma$ is lower bound by $(1-\eps)\var(f(X))$.  We then showed that when $X$ is Gaussian and $Y$ is the output of an additive noise channel then $\sM_\eps(X,Y)\geq 2^{-D(Y)}2^{-2g_\eps(X,Y)}$ for sufficiently small $\eps$, where $D(Y)$ is the non-Gaussianness of $Y$. The significance of this bound is that it gives an operational interpretation for $g_\eps(X,Y)$ in terms of MMSE. Using the approximation obtained for $g_\eps(X,Y)$, we derived a lower bound for $\sM_\eps(X,Y)$ for small $\eps$ which is linear in $\eps$.

\bibliographystyle{IEEEtran}
\bibliography{bibliography}

\appendix
\section{Connection Between Mutual Information and Non-Gaussianness} \label{Appendix A}
For any pair of random variables $(U, V)$ with $I(U; V)\fin$, let $P_{V|U}(\cdot|u)$ be the conditional density of $V$ given $U=u$. Then, we have
\begin{eqnarray}
  I(U; V) &=& \E_{UV}\left[\log\frac{P_{V|U}(V|U)}{P_V(V)}\right]\nonumber\\
  &=&\E_{UV}\left[\log\frac{P_{V|U}(V|U)}{P_{V_\sG|U_\sG}(V|U)}\right]+\E_{UV}\left[\log\frac{P_{V_\sG|U_\sG}(V|U)}{P_{V_\sG}(V)}\right]-\E_{UV}\left[\log\frac{P_{V}(V)}{P_{V_\sG}(V)}\right] \nonumber\\
   &=& I(U_{\sG}; V_\sG) +D(V|U)-D(V),\label{decompositon_MI}
\end{eqnarray}
where $(U_\sG, V_\sG)$ is a pair of Gaussian random variable having the same means, variances and correlation coefficient as $(U,V)$, and $P_{V_\sG|U_\sG}(\cdot|u)$ is the conditional density of $V_\sG$ given $U_\sG=u$, and the quantity $D(V|U)$  is defined in \eqref{Non-Gaussianness_Conditional}.
Replacing $U$ and $V$ with $X$ and $Z_\gamma$, respectively, the decomposition \eqref{decompositon_MI} allows us to conclude that
\begin{equation*}
  I(X; Z_\gamma)=I(X_{\sG}; \sqrt{\gamma}Y_{\sG}+N_{\sG})+D(Z_{\gamma}|X)-D(Z_{\gamma}),
\end{equation*}
and therefore, if $Y=Y_\sG$ is Gaussian, we have
$$I(X; Z_\gamma)=I(X_{\sG}; Z_\gamma)+D(Z_\gamma|X)\geq I(X_{\sG}; Z_\gamma),$$
from which we conclude that when $Y$ is Gaussian then $I(X; Z_\gamma)\leq \eps$ implies that $I(X_\sG; Z_\gamma)\leq \eps$ and hence $g_\eps(X, Y_\sG)\leq g_\eps(X_\sG, Y_\sG)$.

\section{Completion of the Proof of Theorem~\ref{theorem_G_ENSR_ALMOST}}\label{Appendix_Proof_Lemma}
\begin{lemma}\label{Lemma_Vanishing_Conditional_NONG}
    For Gaussian $X_\sG$ and absolutely continuous $Y$ with unit variance, we have
    $$D(Z_\gamma|X_\sG)\leq\frac{\gamma}{2}\left[\mmse(Y_\sG|X_\sG)-\mmse(Y|X_\sG)\right]+o(\gamma).$$
  \end{lemma}
\begin{proof}
Let $E$ be an auxiliary random variable defined as $$
E=
\begin{cases}
1,~~~|Y|\leq L\\
0,~~~\text{otherwise},
\end{cases}
$$
for some real number $M>0$. Note that
\begin{eqnarray}
  \hspace{-0.5cm}D(Z_\gamma|X_\sG=x) &=&h(\sqrt{\gamma}Y_\sG+N_\sG|X_\sG=x)-h(Z_\gamma|X_\sG=x) \nonumber \\
  &\leq& h(\sqrt{\gamma}Y_\sG+N_\sG|X_\sG=x)-h(Z_\gamma|X_\sG=x, E)\nonumber\\
   &=& \frac{1}{2}\log(2\pi e(1+\gamma\var(Y_\sG|X_\sG=x)))\nonumber\\
   &&-\Pr(E=1)h(Z_\gamma|X_\sG=x, E=1)-\Pr(E=0)h(Z_\gamma|X_\sG=x, E=0) \nonumber\\
   &\stackrel{(a)}{\leq}&  \frac{1}{2}\log(2\pi e(1+\gamma\var(Y_\sG|X_\sG=x))-\Pr(E=0)h(N_\sG)\nonumber\\
   &&-\Pr(E=1)h(Z_\gamma|X_\sG=x, E=1)\label{Proof_Lemma_NonGaussianess}
\end{eqnarray}
where $(a)$ follows from the fact that $h(Z_\gamma|X_\sG=x, E=0)\geq h(N_\sG)$.

Prelov \cite{Prelov_Almost_Gaussian} showed that for any random variable $Y$ such that
\begin{equation}\label{Pinsker_Condition}
  \E[|Y|^{2+\alpha}]\leq K\fin,
\end{equation} for some $\alpha>0$, then
\begin{equation}\label{Pinsker}
  h(\sqrt{\gamma}Y+N_\sG)=\frac{1}{2}\log(2\pi e)+\frac{\var(Y)}{2}(\gamma+o(\gamma)),
\end{equation}
where $o(\gamma)$ term depends only on $K$.
Since $Y|\{E=1\}$ satisfies \eqref{Pinsker_Condition}, we can use \eqref{Pinsker} to evaluate $h(Z_\gamma|X_\sG=x, E=1)$ in \eqref{Proof_Lemma_NonGaussianess} which yields
\begin{eqnarray}
D(Z_\gamma|X_\sG=x)&\leq& \frac{1}{2}\log(2\pi e(1+\gamma\var(Y_\sG|X_\sG=x))-\Pr(E=0)\frac{1}{2}\log(2\pi e)\nonumber\\
&&-\Pr(E=1)\left[\frac{1}{2}\log(2\pi e)+\frac{\var(Y|X_\sG=x, E=1)}{2}(\gamma+o(\gamma))\right]\nonumber\\
&=&\frac{1}{2}\log(1+\gamma\var(Y_\sG|X_\sG=x))\nonumber\\
&&-\frac{\var(Y|X_\sG=x, E=1)}{2}(\gamma+o(\gamma))\Pr(E=1).  \label{Pinsker_Outcome}
\end{eqnarray}
Note that since $\var(Y)\fin$ and $X_\sG$ has a positive density, $\var(Y|X_\sG=x)\fin$ for almost all $x$ (except for $x$ in a set of zero Lebesgue measure). Hence, we can choose $L$ sufficiently large such that for any given $\delta>0$,
$$\Pr(E=1)\geq1-\delta,$$ and
$$\var(Y|X_\sG=x, E=1)\geq \var(Y|X_\sG=x)-\delta.$$ Therefore, invoking the inequality $\log(1+u)\leq u$ for $u>0$, we can write
$$D(Z_\gamma|X_\sG=x)\leq \frac{\gamma}{2}\left[\var(Y_\sG|X_\sG=x)-(\var(Y|X_\sG=x)-\delta)(1-\delta)\right]+o(\gamma),$$
from which and the fact the $\delta$ is arbitrarily small the result follows. \qed
\end{proof}
%

\end{document}